\documentclass[11pt]{article}
\usepackage[a4paper,margin=1in]{geometry}
\usepackage{amsmath,amssymb,amsthm,mathtools,microtype,comment,authblk,enumitem}
\usepackage[T1]{fontenc}
\usepackage{lmodern}
\usepackage[hidelinks,hypertexnames=false]{hyperref}
\usepackage{needspace}
\usepackage{setspace}
\usepackage{thmtools}

\newcommand{\ACzerop}{{\mathrm{AC}^0[p]}}
\newcommand{\FPHP}{{\mathrm{FPHP}}}
\newcommand{\PHP}{{\mathrm{PHP}}}
\newcommand{\ontoPHP}{{\mathrm{onto}\text{-}\mathrm{PHP}}}
\newcommand{\bitPHP}{{\mathrm{bit}\text{-}\mathrm{PHP}}}

\usepackage{ifthen}
\newboolean{personalsite}
\setboolean{personalsite}{true}

\makeatletter
\renewcommand\thmt@autorefsetup{%
  \@xa\def\csname\thmt@envname autorefname\@xa\endcsname\@xa{\thmt@thmname}%
}
\makeatother

\newtheorem{theorem}{Theorem}[section]
\newtheorem{lemma}[theorem]{Lemma}
\newtheorem{fact}[theorem]{Fact}
\newtheorem{proposition}[theorem]{Proposition}
\newtheorem{corollary}[theorem]{Corollary}
\newtheorem{remark}[theorem]{Remark}
\theoremstyle{definition}\newtheorem{definition}[theorem]{Definition}
\newcommand{\Fp}{\mathbb F_p}
\newcommand{\PC}{\mathrm{PC}}
\newcommand{\Bool}{\mathrm{Bool}}
\newcommand{\ML}{\operatorname{ML}}
\newcommand{\Span}{\operatorname{span}}
\newcommand{\rank}{\operatorname{rank}}
\newcommand{\Q}{\mathcal Q}
\newcommand{\V}{\mathcal V}
\newcommand{\C}{\mathcal C}
\newcommand{\mult}{\operatorname{mult}}
\newcommand{\MOD}{\operatorname{MOD}}

\newcommand{\Res}{\operatorname{Res}}

\title{A Degree--Size Relation for Resolution over Polynomials}

\author{Shuo Pang}
\affil{University of Bristol, UK\\\texttt{s.pang@bristol.ac.uk}}
\date{}

\begin{document}
\maketitle

\vspace{10pt} 

\begin{abstract}
For every constant-width CNF, we show that linear degree in polynomial calculus (PC) implies exponential size in resolution over constant-degree polynomials, over the same prime field.

Applications include exponential lower bounds in $\Res(\PC_r/\Fp)$ and hence in $\Res(\oplus_p)$, separations between different moduli, improved lower bounds for $\operatorname{Res}(k)$ up to $k=\varepsilon\log n$, proof-search consequences, and an implication of superpolynomial $\ACzerop$-Frege bounds from very strong PC degree lower bounds.

The proof constructs a Razborov--Smolensky approximation and, inspired by \cite{Braun2026exponential}, a polynomial whose multiplication yields inference-preservation. We show their joint existence by measuring the approximation errors by the ranks of multiplication maps induced by the error-witness polynomials, modulo bounded-degree PC consequences.
\end{abstract}
\vspace{30pt} 

\ifthenelse{\boolean{personalsite}}{
\paragraph{AI Disclosure}
The author studied \cite{Braun2026exponential} by asking targeted questions with ChatGPT-6 Extra High. He then supplied a simplified proof architecture to ChatGPT-6 Pro, together with a suggestion to use a multiplier to restore inferences disrupted by Razborov--Smolensky.

The initial exchanges identified obstructions to certain approaches for characteristics other than $2$. After several more, the model proposed a lower-bound proof for $\bitPHP$ in $\Res(\oplus_p)$. The author subsequently suggested seeking a measure of the errors relative to the input formula. The model developed this into a rank measure modulo bounded-degree PC, and proposed essentially the arguments in Section~\ref{sec:rank}. 

Clearly the suggestions directed the search materially; only their necessity remains unclear. 
The author studied, anatomized and reconstructed the mathematics, then wrote the paper with AI assistance at that stage limited to editorial suggestions.
\vspace{30pt} 

\paragraph{Acknowledgements} 
I thank Jonas Conneryd and Yassine Ghannane who brought \cite{Braun2026exponential} to my attention, and Jonas Conneryd for suggesting the question in characteristic 0.
}
{}

\thispagestyle{empty} 
\clearpage           

\setcounter{page}{1}  

\section{Introduction and main result}
Let $p$ be a fixed prime. 
Let $F$ be a width-$w(F)$ unsatisfiable CNF on $n$ Boolean variables. 
Let 
\[\Q(F)\subseteq\Fp[x]=\Fp[x_1,\ldots,x_n]\] 
be the polynomial-system encoding of $F$, which collects all Boolean axioms $x_i^2-x_i$ and the \emph{falsity polynomials} $Q_J$ of each clause $J\in F$, i.e., $Q_{x_i}=1- x_i$, 
$Q_{\neg x_i} = x_i$, and  
\[Q_J = \prod_{j=1}^m Q_{l_j}\quad \text{for $J=\bigvee_{j=1}^m l_j$}.\] 
So the common zeros of $\Q(F)$ are exactly the satisfying assignments of $F$. 

Write $\Q\vdash_D f$ if $f$ has a degree-$D$ \emph{polynomial calculus} (PC) derivation from $\Q$ \cite{CEI96}: a sequence of polynomials of degree at most $D$, ending with $f$, where each line is either in $\Q$ or obtained from earlier lines by a linear combination $\alpha g+\beta h$ ($\alpha,\beta\in\Fp$) or by multiplication $x_i g$ with a variable. 
Let
\begin{equation}\label{eq:CD}
 \C_D(\Q)=\{f:\Q\vdash_D f\},
\end{equation}
which is a linear subspace of $\Fp[x]_{\le D}$.

The proof system $\Res(\PC_r/\Fp)$ was introduced by Kraj\'i\v{c}ek over $\mathbb F_2$ and subsequently studied over general finite rings by Khaniki~\cite{Krajicek18,Khaniki}. Its lines are disjunctions of degree-$\leq r$ polynomial equations. 
We use the standard (syntactic) deduction rules: 
\emph{syntactic weakening}, 
\emph{contraction}, 
\emph{deletion of false constant equations}, 
\emph{linear combination}
\[
 \frac{A\vee(f=0)\qquad B\vee(g=0)}
      {A\vee B\vee(\alpha f+\beta g=0)},
\]
and 
\emph{multiplication} $A\vee(f=0)\mapsto A\vee(fg=0)$ where all displayed polynomials have degrees $\le r$. 
The axioms $(0=0)$ and $(x_i=0)\vee(x_i-1=0)$ are supplied by default.
A refutation is a finite DAG ending with the empty clause. Its \emph{size} $S$ is the number of nodes in the DAG. 

\begin{restatable}[Degree--size]{theorem}{quantitative}\label{thm:quantitative}
Suppose $\Q(F)\not\vdash_D 1$, where
\[
 D\ge\max\{32(p-1)^2r,\,8w(F)\}.
\]
Then every $S$-node $\Res(\PC_r/\Fp)$ refutation of $F$ satisfies
\begin{equation}\label{eq:main}
 \ln S\ge\frac{D}{32p^2r}\left(\frac{D}{16n}\right)^{(p-1)r}.
\end{equation}
\end{restatable} 
For constant $p,r$ and $w(F)$, linear PC degree lower bounds therefore imply exponential size. 
As applications, we give a few such size lower bounds in Section~\ref{sec:applications}.
\smallskip

Previously, Khaniki~\cite{Khaniki} proved a size--width relation of the form $w=O(\sqrt{(n+S)\ln S})$ which yielded almost quadratic size lower bounds. 

\begin{remark}[Polynomial-clause axioms]\label{rmk:poly-clause}
Theorem~\ref{thm:quantitative} extends to polynomial equation systems $F=\{F_1,\ldots,F_m\}$ where $F_i=\bigvee_{j=1}^{w_i}(H_{i,j}=0)$ and $\deg H_{i,j}(x)\le r$. 
In this case, 
\begin{equation}\label{eq:poly-Qi}
\Q(F)=\{\text{Boolean axioms}\}\cup\{Q_{F_1},\ldots,Q_{F_m}\}, \ \  Q_{F_i}=\ML\left(\prod_{j=1}^{w_i}H_{i,j}\right).
\end{equation} 
Here, $\ML$ stands for multilinearization. And we replace $w(F)$ in the theorem by 
\begin{equation}\label{eq:WF}
W(F)=(p-1)\max_i\sum_{j=1}^{w_i}\deg H_{i,j}.
\end{equation}
\end{remark}

An immediate corollary is for $\Res(\oplus_2)=\Res(\PC_1/\mathbb{F}_2)$, the system operating over disjunctions of $\mathbb F_2$-affine equations \cite{IS14}.\footnote{This system is sometimes defined to also allow \emph{semantic weakening}, where a disjunction of $\mathbb F_2$-affine equations can be derived in a step if it is implied by an earlier line over $\{0,1\}^n$. This does not affect the discussions here since the rule can be simulated syntactically with polynomial overhead, cf. \cite{IS14}.}
 
\begin{corollary}[$\Res(\oplus_2)$]\label{cor:Res2}
Suppose $F=\{F_1,\ldots,F_m\}$ where each $F_i(x)$ is a disjunction of at most $w$ $\mathbb F_2$-affine equations. 
If $\Q(F)\not\vdash_{D}1$ where $D$ satisfies $D\geq 32w$, then every $\mathrm{Res}(\oplus_2)$ refutation of $F$ has size
\begin{equation}\label{eq:mainbound}
 S\ \ge\ \exp\left( \Omega\Big( D ^2/n\Big)\right).
\end{equation}
\end{corollary}
In this case, the parameter dependence is comparable to classical width/degree--size relations \cite{B-SW,IPS}. 
Prior to~\cite{Braun2026exponential}, superpolynomial lower bounds beyond tree-like $\Res(\oplus_2)$ began with the (bottom-)regular lower bounds of Efremenko, Garl\'ik, and Itsykson~\cite{EGI25}, followed by separations between regular and general proofs~\cite{BCD24}. Subsequent work proved lower bounds under depth restrictions reaching nearly quadratic~\cite{AI25,EI25,BBCI26,IPS26}, and supercritical size--depth tradeoffs~\cite{IK26,IPS26}.

\paragraph{Proof overview.} 
For a proof line $L=(h_{L,1}(x)=0)\vee\cdots\vee(h_{L,m_L}(x)=0)$, let \(g_{L,i}\) denote the natural \(\{0,1\}\)-valued truth indicator
of its \(i\)th predicate, i.e., $g_{L,i}(x)=1 \iff h_{L,i}(x)=0$. 
The falsity polynomial of the line is
\begin{equation} 
\widehat L=\prod_{i=1}^{m_L}(1-g_{L,i}),
\end{equation}
where the equation \(\widehat L=0\) represents the predicate \(L\).

We associate with \(L\) a Razborov--Smolensky polynomial
\begin{equation}
P_L=\prod_{j=1}^t(1-q_{L,j}^{p-1}),
\qquad
q_{L,j}\in\Span\{g_{L,i}:i\in[m_L]\}.
\end{equation}
This is a one-sided approximation to $\widehat L$, where an error occurs only when \(P_L=1\) while some \(g_{L,i}=1\), and such an error is witnessed by one of the \emph{error-witnesses} \(g_{L,i}P_L\) being 1.\footnote{The use of Razborov--Smolensky approximation in proof complexity dates back to~\cite{BIKPRS96}, where it was implemented differently using extension variables.}

Such a line-by-line approximation however need not respect inference.
To repair this, we use a common-multiplier idea extracted from \cite{Braun2026exponential}. 
Namely, for any polynomial $f$, the $f$-multiplied lines $f\widehat L=0$ has zero set that represents the ORified line
\begin{equation}
 (f=0)\vee L
\end{equation}
so they preserve inferences. There, $f$ was chosen so that $f\widehat L$ vanishes on every high-rank line $L$.

We instead choose $f$, $\{P_L\}$ jointly so that $f$ annihilates the error-witnesses, not the lines themselves, and only so modulo bounded-degree PC consequences of the formula. 
Concretely, we impose
\begin{equation}
f g_{L,i}P_L\in\C_D
\quad\text{for every }L,i,
\quad\text{while $f\notin\C_D$}.
\end{equation}
To show the existence of such $f$ and $P_L$, we measure the error set by the rank of the map
\begin{equation}\label{eq:intro-rank}
f\longmapsto \Bigl( g_{L,i} P_L f \mod \C_D\Bigr)_{i=1}^{m_L}\quad \text{defined on $\Fp[x]_{\leq K}$},
\end{equation}
$K\le D$ a degree parameter. The main algebraic tool (Lemma~\ref{lem:pencil} and Lemma~\ref{lem:onestep}) shows that by adding one more approximation factor to $P_L$, we can shrink this rank by a constant factor. 
So by increasing the approximation level $t$ to a proper value, the sum of ranks \eqref{eq:intro-rank} over all lines is smaller than $\dim(\V_K/(\V_K\cap \C_D))$, at which point there must exist an $f\notin\C_D$ annihilating all the error-witnesses.
\smallskip 

The paper is organized as follows. 
Section~\ref{sec:rank} proves the linear-algebraic lemmas. 
Section~\ref{sec:approximation} uses them to establish the existence of approximators that allow a nontrivial multiplier. 
Section~\ref{sec:simulation} simulates the resulting derivation by low-degree PC. 
Section~\ref{sec:bounds} computes the parameters, and Section~\ref{sec:applications} gives some applications. 
We conclude the paper in Section~\ref{sec:conclusions} with a generalization to $\ACzerop$-Frege and open problems.

\section{Preliminaries}\label{sec:preliminaries}
We write $\log$ for base-2 logarithm. 
Denote $\mathbb F_p[x]=\mathbb F_p[x_1,\ldots,x_n]$ and
\[\V_k=\mathbb F_p[x]_{\le k}.\]
All degree parameters are nonnegative integers, and $r\ge1$ in $\Res(\PC_r/\Fp)$. 
For a set $\mathcal A$ of polynomial equations, we identify an equation $P=0$ with the polynomial $P$ and write $\mathcal A\vdash_k P$ for a PC derivation whose polynomials have degree at most $k$. 

We denote an ordinary clause $J=l_1\vee\ldots\vee l_s$ as the set $\{l_1,\cdots,l_s\}\subseteq \{x_i,1-x_i:i\in[n]\}$. Recall its falsity polynomial is $Q_J=\prod_{i=1}^s(1-l_i)$. 
A CNF $F$ encoded as a set of polynomial constraints is 
\[\Q(F)=\Bool\cup\{Q_J:J\in F\},\qquad \Bool=\{x_i^2-x_i:i\in[n]\}.\] 
We use $w=w(F)$ for the width of $F$.

Fix a degree parameter $D$ and input polynomial set $\Q\supseteq \Bool$ throughout Sections~\ref{sec:preliminaries}--\ref{sec:simulation}. 
Denote by $\C_D=\C_D(\Q)$ the space of all polynomials derivable from $\Q$ in PC degree $D$. 
Write 
\[
 \C_{\le k}=\C_D\cap\V_k\quad(0\le k\le D)
\]
where note that the inclusion $\C_k(\Q)\subseteq \C_{\leq k}$ may be proper. 
For $f\in \V_D$, write 
\begin{equation}\label{eq:modulo}
[f]\ \ \text{for its class in the quotient linear space $\mathcal V_D/\C_D$}.
\end{equation}

The following facts are standard and immediate. 

\begin{fact}\label{fact:multiplication}
If $\mathcal A\vdash_D f$, then
$\mathcal A\vdash_{\max\{D,\deg f+\deg g\}} fg$ for every polynomial $g$.
Consequently,
\begin{equation}\label{eq:absorption}
 f\in\C_D,\quad \deg f+\deg g\le D
 \quad\Longrightarrow\quad fg\in\C_D.
\end{equation}
\end{fact}

\begin{fact}\label{fact:boolean}
Every polynomial $f(x)$ in Boolean variables has a unique multilinear representative $\ML(f)$ which is functionally the same as $f: \{0,1\}^n \to \Fp$. Moreover,
\[\deg\ML(f)\le\deg f, \quad \Bool\vdash_{\deg f}f - \ML(f).\] 
\end{fact}

For $g\in\Fp[x]$ with $\deg g\le d$, it follows from Fact~\ref{fact:boolean} that
\begin{equation}\label{eq:frobbool}
 \Bool\vdash_{pd}g^p-g.
\end{equation}

\begin{lemma}[Base change]\label{lem:basechange}
Let $E/\Fp$ be a field extension, and regard $\C_D^E:=E\otimes_{\Fp}\C_D$ as an $E$-subspace of $E[x]_{\leq D}$. 
Then for every $k\le D$,
\[
E\otimes_{\Fp}(\C_D\cap\V_k) = \C_D^E\cap E[x]_{\le k}.
\]
Moreover, if $f\in\C_D^E$, \(g\in E[x]\), and $\deg f+\deg g\le D$, then $fg\in\C_D^E$.
\end{lemma}

\begin{proof}
The first equality is the finite-dimensional fact $E\otimes (V\cap W)= (E\otimes V)\cap(E\otimes W)$ applied to $V:=\mathcal C_D, W:=\mathcal V_k$. 

For the moreover part, if $\deg f=j$ then $f\in E\otimes_{\Fp}(\C_D\cap\V_j)$ by the first part, so $f=\sum_\alpha \lambda_\alpha f_\alpha$ where $f_\alpha\in\C_D\cap\V_j$ and $\lambda_\alpha\in E$.
Write $g=\sum_\beta\mu_\beta x^\beta$, $\mu_\beta\in E$, then by Fact~\ref{fact:multiplication}, every $f_\alpha x^\beta$ belongs to $\C_D$ since its degree is at most $D$. Hence, $fg=\sum_{\alpha,\beta}\lambda_\alpha\mu_\beta f_\alpha x^\beta\in\C_D^E$.
\end{proof}

\section{Algebraic tool: rank retention}\label{sec:rank}
This section proves a rank-preservation bound for linear combinations of multiplication maps (Corollary~\ref{cor:descent}).
 
Whenever they occur below, $\lambda_1,\ldots,\lambda_s$ denote algebraically independent indeterminates over $\mathbb F_p$. 

The following lemma bounds the rank loss when a matrix with $\lambda$-coefficients is specialized to values in $\Fp$. 
It relies on a multiplicity Schwartz--Zippel lemma in \cite{DKSS}.

\begin{lemma}[Rank specialization]\label{lem:pencil}
Let $M_1,\ldots,M_s$ be matrices over $\Fp$ of the same dimensions. For algebraically independent indeterminates
$\lambda_1,\ldots,\lambda_s$, put
\[
 M(\lambda)=\sum_{i=1}^s\lambda_iM_i,
 \qquad R=\rank_{\Fp(\lambda)}M(\lambda).
\]
Then 
\begin{equation}\label{eq:pencil}
\mathbb E_{a\in\Fp^s}\Big[ \rank M(a) \Big] \ge(1-\frac{1}{p})R.
\end{equation}
\end{lemma}
\begin{proof}
For a nonzero polynomial\footnote{meaning nonzero in the polynomial ring $\Fp[z_1,\ldots,z_s]$, not as a function on $\Fp^s$.} 
$H(z_1,\ldots,z_s)$ and a point $a\in\Fp^s$, let $\mult(H,a)$ be the \emph{least} degree of a nonzero term of $H(a+z)$.
We use a multiplicity bound~\cite[Lemma~8]{DKSS}
\begin{equation}\label{eq:multiplicity}
 \sum_{a\in\Fp^s}\mult(H,a)\le(\deg H)p^{s-1}.
\end{equation}
For completeness, we include a proof. 
For $s=1$, it is the bound on the sum of root multiplicities. 
For $s>1$, write $H(y,z)=\sum_{j=0}^t H_j(y)z^j$ with $H_t\ne0$. Fix $b\in\Fp^{s-1}$ and let $\mu:=\mult(H_t,b)$.
Choose a monomial $y^\alpha$ of degree $\mu$ appearing in $H_t(b+y)$, then its coefficient in $H(b+y,z)$, denoted $h_\alpha(z)$, has degree $t$.
Therefore for all $c\in\Fp$,
$\mult(H,(b,c))\le\mu+\mult(h_\alpha,c)$.
Summing over $c$ gives at most $p\mu+t$, and summing over $b$ and applying induction to $H_t$ gives~\eqref{eq:multiplicity}.

If $R=0$, the assertion \eqref{eq:pencil} is immediate. Otherwise choose an $R\times R$ submatrix of $M(\lambda)$ with nonzero determinant $H(\lambda)$. 
Given point $a\in \Fp^s$, let the rank of this submatrix be $\rho_a \leq \rank M(a)$. 
By using invertible row operations over $\Fp$, we can make $R - \rho_a$ rows vanish at $a$. This means, by expanding the sum-product definition of determinant, that the terms in $H(\lambda + a)$ all have degree at least $R - \rho_a$. Consequently,
\[
 R-\rank M(a) \le  R - \rho_a \le\mult(H,a).
\]
Average and apply~\eqref{eq:multiplicity}.
\end{proof}

The next is an elementary Frobenius-descent fact for linear subspaces. 

\begin{lemma}\label{lem:subspace-descent}
Let $E$ be a field of characteristic $p$ and let $W\subseteq E^m$
be an $E$-linear subspace. For $u=(u_1,\ldots,u_m)$, write
\[u^{[p]}=(u_1^p,\ldots,u_m^p).\]
If $u^{[p]}\in W$ for every $u\in W$, then
\[
 W=E\otimes_{\Fp}(W\cap\Fp^m),
\]
under the natural identification inside $E^m$.
\end{lemma}

\begin{proof}
Let $T$ be the reduced row-echelon matrix whose nonzero rows form a basis of $W$. The rows of $T^{[p]}$ belong to $W$ and have the same pivot columns as $T$, so they are linearly independent. 
Thus $T^{[p]}$ is also a reduced row-echelon basis matrix of $W$. Uniqueness gives $T^{[p]}=T$. Every entry therefore satisfies $c^p=c$ and belongs to $\Fp$. Hence $W$ has a basis in $\Fp^m$, which proves the assertion.
\end{proof}

Let $E=\Fp(\lambda_1,\ldots,\lambda_s)$ be the field extension. By the right-exactness of tensoring,
\[
E\otimes_{\Fp}(\V_D/\C_D)
   \cong E[x]_{\le D}/\C_D^E
\qquad \text{where}\ \ 
\C_D^E:=E\otimes_{\Fp}\C_D.
\]
We will use this identification in the rest of the section. Below, the only properties of $\C_D$ that we need are that it includes all Boolean-axiom polynomials $x_i^2-x_i$, and that it is multiplication-absorptive within degree $D$ (Fact~\ref{fact:multiplication}).

\begin{lemma}[Kernel descent]\label{lem:generic-descent}
Let $M_1,\ldots,M_s:U\to W$ be linear between finite-dimensional $\Fp$-spaces $U$ and $W$. 
Set
\[
 M(\lambda)=\sum_{i=1}^s\lambda_iM_i,\qquad
 M(\lambda^p)=\sum_{i=1}^s\lambda_i^pM_i
\]
where the maps are extended linearly over $E$. 
If
\begin{equation}\label{eq:Frob-cond}
 \ker_E M(\lambda)\subseteq\ker_E M(\lambda^p)
\end{equation}
then
\begin{equation}\label{eq:kernelgeneric}
 \ker_E M(\lambda) = E\otimes_{\Fp}\bigcap_{i=1}^s \ker M_i.
\end{equation}
In that case, for the stacked map $\overline{M}:U\to W^s$ defined by $\overline{M}u=(M_1 u,\ldots,M_s u)$, it holds that
\begin{equation}\label{eq:rank-nullity}
 \rank_E M(\lambda)=\rank_{\Fp} \overline{M}.
\end{equation}
Consequently, there exists $c\in\Fp^s$ that satisfies

\begin{equation}\label{eq:goodpencil}
 \rank_{\Fp}\Bigl( \sum_{i=1}^s c_i M_i \Bigr) \ge ( 1- \frac1p )\rank_{\Fp}\overline{M}.
\end{equation}
\end{lemma}

\begin{proof}
Choose $\Fp$-bases of $U, W$ so the maps $M_i$ have matrix form over $\Fp$. 
Each minor of $M(\lambda^p)$ is the $p$th power of the corresponding minor of $M(\lambda)$, so the two matrices have equal ranks. 
The assumed kernel inclusion \eqref{eq:Frob-cond} is therefore an equality.

Denote $K=\ker_E M(\lambda)$. If $u\in K$, raising $M(\lambda)u$ to the $p$th power coordinatewise gives $M(\lambda^p)u^{[p]}=0$, so by the kernel equality, $u^{[p]}\in K$, and thus by Lemma~\ref{lem:subspace-descent}, $K$ has a basis over $\Fp$.
For $u\in K\cap U$, the identity $\sum_i\lambda_iM_i u=0$ implies $M_i u=0$ for every $i$ since the $\lambda_i$'s are algebraically independent.
Hence, $K\cap U=\bigcap_i\ker M_i$, proving \eqref{eq:kernelgeneric}. Consequently, $\rank_E M(\lambda) = \dim U - \dim(\mathop\cap_i \ker M_i) = \rank_{\Fp} \overline{M}$, proving \eqref{eq:rank-nullity}. 
Then Lemma~\ref{lem:pencil} gives~\eqref{eq:goodpencil}.
\end{proof}

\begin{remark}\label{rmk:Frob-cond}
Without condition \eqref{eq:Frob-cond} it can happen that $\rank_E M(\lambda) < \rank_{\Fp} \overline{M}.$ 
For example, consider
\[
M_i:\mathbb F_p^s\to\mathbb F_p, \qquad M_i(u)=u_i.
\]
The map $\overline M$ is the identity of rank $s$, while $M(\lambda)$ is $M(\lambda)(u)=\sum_i \lambda_i u_i$ having rank $1$ over $E$. 
For $s=2$, for instance, $(\lambda_2,-\lambda_1)\in\ker M(\lambda)$ while $\ker_{\Fp}(\overline M)=0$.
\end{remark}

We now verify condition \eqref{eq:Frob-cond} for the specific linear operator we later use. 

\begin{corollary}[Kernel descent for multiplication map]\label{cor:descent}
Let $g_1,\ldots,g_s,h_1,\ldots,h_t\in\Fp[x]$ satisfy
$\deg g_i\le a$ and $\deg h_j\le b$.
Suppose $k\ge0$ and
\begin{equation}\label{eq:descentbudget}
 b+k+pa\le D.
\end{equation}
Define
\begin{equation}\label{eq:Mi}
 M_i:\V_k\longrightarrow(\V_D/\C_D)^t,
 \qquad
 M_i f=([f g_i h_j])_{j=1}^t.
\end{equation}
For $E=\Fp(\lambda_1,\ldots,\lambda_s)$, these maps satisfy
\begin{equation}\label{eq:kernel-verify}
 \ker_E M(\lambda)\subseteq\ker_E M(\lambda^p).
\end{equation}
In particular, the rank-preservation bound \eqref{eq:goodpencil} holds.
\end{corollary}

\begin{proof}
Denote $q_\lambda=\sum_i\lambda_i g_i$.
Any $f\in\ker_E M(\lambda)$ satisfies $h_\alpha q_\lambda f\in\C_D^E$, so $h_\alpha q_\lambda^p f\in\C_D^E$ by the moreover part of Lemma~\ref{lem:basechange} and \eqref{eq:descentbudget}. 
Recall that $g_i\in \Fp[x]$, so $g_i^p-g_i\in \C_D$ by \eqref{eq:frobbool}. 
Together, 
\begin{equation}\label{eq:multiplication-check-1} 
h_\alpha q_\lambda^p f  -  \sum_i\lambda_i^p h_\alpha (g_i^p-g_i) f  \in\C_D^E.
\end{equation}
Since $q_\lambda^p = \sum_{i} \lambda_i^p g_i^p$ over characteristic $p$, the element in \eqref{eq:multiplication-check-1} is $\sum_i\lambda_i^p h_\alpha g_i f$. So $f\in\ker_E M(\lambda^p)$, and Lemma~\ref{lem:generic-descent} gives the bound.
\end{proof}

\section{Approximating a proof}\label{sec:approximation}
Fix any $g_1,\ldots,g_s\in\Fp[x]$ of degree at most $d$. 
We consider polynomials, for integers $t\ge1$,
\begin{equation}\label{eq:product}
 P=\prod_{j=1}^t(1-q_j^{p-1}),
 \qquad q_j\in\Span_{\Fp}\{g_1,\ldots,g_s\}.
\end{equation}
On the Boolean cube, $P$ is 0-1 valued. The equation $P=0$ is a one-sided approximation of the predicate $\mathop\lor_{i=1}^s (g_i \neq 0)$ \cite{Razborov-polynomial,Smolensky}, since $P(x)=1$ on the common zeros of $g_i$, and errors occur only where $P(x)=1$ while some $g_i(x)\ne0$. 
Thus all errors are witnessed by one of the \emph{error polynomials} $Pg_i$ being nonzero. 
We measure the dimension of the common mod-$\C_D$ annihilators of these polynomials, as follows. 

For parameters $a,k$ where $a+k +d \le D$, and any $P\in \mathcal V_a$, define the \emph{error-quotient} map:
\begin{equation}\label{eq:rankdef}
 \begin{aligned}
 T_{P,k}:\V_k&\longrightarrow(\V_D/\C_D)^s\\ 
       T_{P,k}(f)&=([Pg_if])_{i=1}^s 
 \end{aligned}
\end{equation}
and let 
\begin{equation}\label{eq:ZPcP}
Z_P(k) = \ker T_{P,k}, \qquad c_P(k) =\rank T_{P,k}.
\end{equation}
Thus $c_P(k)$ counts the independent conditions that multiplication by $f$ must satisfy to put every error polynomial in $\C_D$. 

We need a degree-based inequality between these ranks. 
Recall the dimension of the multilinear polynomials of degree at most $k$ ($\leq n$) is
\begin{equation}\label{eq:boolean-dimension}
 N_n(k)=\sum_{j=0}^{k}\binom nj.
\end{equation}
The following inequality uses basics of degree-respecting monomial orders and a simple double-counting argument (cf. \cite[Prop. 4.5.1]{EngelSperner}). 

\begin{lemma}\label{lem:hilbert}
Let $a+K+d\le D$ and $K\le n$. Then the normalized rank $\frac{c_P(k)}{N_n(k)}$ is nonincreasing for $0\le k\le K$. In particular, if $K\ge d$, then
\begin{equation}\label{eq:hilbert}
 c_P(K-d) \ge \frac{N_n(K-d)}{N_n(K)} {c_P(K)}.
\end{equation}
\end{lemma}
\begin{proof}
Choose a degree-respecting monomial order, such as the graded lexicographic order. 
The monomials of degree $\leq K$ that are not leading monomials of elements of $Z_P(K)$ form a quotient basis of $\V_K/Z_P(K)$. 
These elements are closed under sub-monomials since $x^\alpha Z_P(k)\subseteq Z_P(k+|\alpha|)$ when $k+|\alpha|\le K$, and they are multilinear since all the Boolean axioms lie in \(Z_P(k)\) if $k\geq 2$. Moreover, since $Z_P(k)=Z_P(K)\cap\V_k$, they are degree-compatible, i.e., the elements of degree at most $k$ form a basis of $\V_k/Z_P(k)$ for $k\le K$.

Let $b_j$ count these basis monomials of degree $j$. By a double-counting of the containment relations between consecutive degrees, we have $(j+1)b_{j+1}\le(n-j)b_j$, thus $b_j/\binom nj$ is nonincreasing, and so are their cumulative averages, which are $c_P(k)/N_n(k)$.
\end{proof}

\begin{lemma}[Approximation gain]\label{lem:onestep}
Suppose $P\in\mathcal V_a$, $K\ge d$ and $a+K+pd\le D$. Then there exists $q\in\Span_{\Fp}\{g_1,\ldots,g_s\}$ satisfying
\begin{equation}\label{eq:shrinkage}
 c_{P(1-q^{p-1})}(K)  \le c_P(K) - \left(1 - \frac{1}{p}\right) c_P(K-d).
\end{equation}
\end{lemma}

\begin{proof}
Define a map on $\V_{K-d}$ by
\begin{equation}\label{eq:Mi-application}
M_i f=([Pg_\alpha g_i f])_{\alpha=1}^s.
\end{equation}
Then the stacked map 
\begin{equation}\label{eq:overline-M}
\overline{M}=(M_1,\ldots,M_s)
\end{equation}
has kernel $Z_P(K-d)$. Indeed, one inclusion follows from the bounded-degree multiplication absorption property of $\C_D$ (Fact~\ref{fact:multiplication}). 
Conversely, if $f\in\ker \overline{M}$, then $Pg_i^2f\in\C_D$ for every $i$, so multiplying by $g_i^{p-2}$ and using $g_i^p-g_i\in\C_D$, we get that $Pg_if\in\C_D$. Thus
\begin{equation}\label{eq:rank-M}
\rank \overline{M}=c_P(K-d).
\end{equation}

We apply Corollary~\ref{cor:descent} with $h_\alpha \leftarrow Pg_\alpha$, $t\leftarrow s$, and domain $\V_{K-d}$, where the degree condition is satisfied since $\deg h_\alpha\le a+d$ and $(a+d)+(K-d)+pd\le D$. As a result, we get a choice 
\begin{equation}\label{eq:choice-q}
\begin{aligned}
q=\sum_i c_i g_i \quad &\text{such that}\\
\rank\Bigl( f \mapsto ([Pg_\alpha qf])_\alpha \Bigr) \ge & \left(1-\frac1p\right)c_P(K-d).
\end{aligned}
\end{equation}
Since the map in \eqref{eq:choice-q} decomposes by $f\mapsto qf \mapsto T_{P,K}(qf)$, this inequality means that $q\mathcal V_{K-d}$ modulo its intersection with $Z_P(K)$ has dimension at least the RHS, or
\begin{equation}\label{eq:onestep-1}
\dim\frac{q\V_{K-d}+Z_P(K)}{Z_P(K)}
\ge
\left(1-\frac1p\right)c_P(K-d).
\end{equation}

For $P'=P(1-q^{p-1})$, its multiplication map on $\V_K$ has kernel 
\begin{equation}\label{eq:onestep-2}
Z_{P'}(K) \supseteq q\V_{K-d}+Z_P(K).
\end{equation}
To see this, first take any $g \in Z_P(K)$. Then $Pg_i g\in\C_D$, so $(1-q^{p-1})Pg_i g\in \C_D$ because its degree is at most $a+K+pd\le D$ and we can use Fact~\ref{fact:multiplication}. This shows $g \in Z_{P'}(K)$. 
Next, take any $g\in\V_{K-d}$, then $P'g_i(q g)=Pg_i g (q-q^p)\in\C_D$ since $q-q^p\in\C_D$, where the degree is at most $a+d+(K-d)+pd=a+K+pd\le D$, so $q g\in Z_{P'}(K)$.

Together, \eqref{eq:onestep-1} and \eqref{eq:onestep-2} give
\[\dim Z_{P'}(K)-\dim Z_P(K) \ge \left(1-\frac1p\right)c_P(K-d).\]
The lemma follows since $c_P(K)=\dim \V_K - \dim Z_P(K)$ and $c_{P'}(K)=\dim \V_K - \dim Z_{P'}(K)$.
\end{proof}

We now use Lemma~\ref{lem:onestep} to construct the approximators and a common-multiplier $f$.

\begin{theorem}[Common multiplier]\label{thm:common}
Suppose $1\notin\C_D$ and the parameters $t,d,K$ satisfy $t\geq 1$, $d\le K \le n$, and $K+d+(p-1)dt \le D$. 
Suppose we are given an $S \ge1$ and for each $\nu\in[S]$, a family of polynomials $(g_{\nu,i})_{i=1}^{m_\nu}$ of degree at most $d$. Denote
\begin{equation}\label{eq:approx-parameters}
 \eta=\left(1 - \frac{1}{p}\right)\frac{N_n(K-d)}{N_n(K)},\qquad
 h_K=\dim \frac{\V_K}{\C_{\le K}}.
\end{equation}
Then for each $\nu$, there exists $P_\nu$ of the form \eqref{eq:product} with $t$ factors, such that the corresponding multiplication-map rank $c_{P_\nu}(K)$ defined by~\eqref{eq:ZPcP} satisfies 
\begin{equation}\label{eq:ratio-t}
c_{P_\nu}(K)\le h_K(1-\eta)^t.
\end{equation}
In particular, if 
\begin{equation}\label{eq:size-cond}
S(1-\eta)^t<1,
\end{equation} 
then there exists a multilinear $f\in\V_K$ satisfying
\begin{equation}\label{eq:relativegoal}
 f\notin\C_D,\qquad fg_{\nu,i}P_\nu\in\C_D
       \quad\text{for all }\nu,i.
\end{equation}
\end{theorem}
\begin{proof}
Fix a family $\nu$. The map $T_{1,K}$ annihilates $\C_{\le K}$ by Fact~\ref{fact:multiplication}, so $\rank T_{1,K}\le h_K$. We apply Lemma~\ref{lem:hilbert} and Lemma~\ref{lem:onestep} $t$ times, where in round $j$ we find a  factor $q_{\nu,j}\in\Span_{\Fp}\{g_{\nu,1},\ldots,g_{\nu,m_\nu}\}$ satisfying \eqref{eq:shrinkage}. 
Before round $t$, the product $\prod_{j<t}(1 - q_{\nu,j}^{p-1})$ has degree at most $(t-1)(p-1)d$, so the degree conditions in the lemmas are satisfied due to $(t-1)(p-1)d+K+pd\le D$.
Each round reduces the rank of the resulting map by a factor at most $1-\eta$ by concatenating \eqref{eq:hilbert} and  \eqref{eq:shrinkage}.

All $S$ final maps are defined on $\V_K$ and eliminate $\C_{\leq K}$, so their joint map has rank at most $S(1-\eta)^t h_K<h_K$, where recall $h_K>0$ since $1\notin\C_D$. 
By dimension counting, there is an $f\in \V_K\backslash \C_{\le K}$ in the common kernel. Its multilinearization satisfies~\eqref{eq:relativegoal}.
\end{proof}

\section{Simulation in PC}\label{sec:simulation}
It should be intuitively clear that by applying Theorem~\ref{thm:common} to a proof $\Pi$, one can get an ``approximated proof'' in low-degree $\PC$. Below we give the syntactic details.\smallskip

For $h\in\Fp[x]_{\leq r}$, let $g_h$ be the Boolean indicator of the predicate $h=0$:
\begin{equation}\label{eq:gh}
 g_h=\ML(1-h^{p-1}).
\end{equation}
I.e., for all $x\in \{0,1\}^n$, $g_h(x)=1$ iff $h(x)=0$. 
For $L=\bigvee_{i=1}^{m_L}(h_{L,i}=0)$, 
write $g_{L,i}=g_{h_{L,i}}$.

Given a size-$S$ $\Res(\PC_r/\Fp)$ refutation $\Pi$ of $F$, let 
\begin{equation}\label{eq:d}
d=(p-1)r
\end{equation}
so $\deg g_h \leq d$.
Assume there are parameters $D,K$ such that the hypotheses of Theorem~\ref{thm:common} hold, including $1\notin\C_D=\C_D(\Q(F))$ and $S(1-\eta)^t<1$.
Applying the theorem to the families $(g_{L,i})_{i=1}^{m_L}$ indexed by proof lines, we get the choice of approximators
\begin{equation}\label{eq:approximator}
 P_L=\prod_{j=1}^t(1-q_{L,j}^{p-1})
\quad\text{where}\ \ 
 q_{L,j}\in\Span_{\Fp}\{g_{L,1},\ldots,g_{L,m_L}\}\text{ for all } j\in[t],
\end{equation}
along with an $f\in\V_K$ such that
\begin{equation}\label{eq:f-errors}
 fP_Lg_{L,i}\in\C_D
 \quad\text{for every }L,i.
\end{equation}
Denote $\Delta=t(p-1)d$, then $\deg P_L \le \Delta$.
 
For the next proof it is convenient to have an expanded form of \eqref{eq:approximator}. 
Assume
\begin{equation}\label{eq:qLj}
 q_{L,j}=\sum_i c_{L,j,i}g_{L,i}.
\end{equation}
Then for $P_L$ from \eqref{eq:approximator}, we have expression
\begin{equation}\label{eq:prefix}
1 - P_L=\sum_i U_{L,i}\cdot g_{L,i}\quad \text{where}\quad \deg U_{L,i}\le (p-2)d+(t-1)(p-1)d=\Delta-d.
\end{equation}
Here the polynomials $U_{L,i}$ can be chosen by telescoping the product in \eqref{eq:approximator}; their expressions are unimportant.

\begin{proposition}[PC simulation]\label{prop:simulation}
Let $\Pi$, $f$, $\{P_L\}$ be as above, where $\deg f\le K$. 
If
\begin{equation}\label{eq:simdegree}
 K+\max\{2\Delta+2d,\Delta+w(F)\}\le D,
\end{equation}
then
\[
 fP_L\in\C_D\ \ \text{for every line $L$ of $\Pi$.}
\] 
In particular, $f\in\C_D$.
\end{proposition}

\begin{proof}
We proceed in topological order on $\Pi$, appending local derivations using the conclusions from previous degree-$\leq D$ derivations. 
By Fact~\ref{fact:multiplication}, it suffices to bound the degree of each such local derivation, henceforth referred to as \emph{displayed degrees}.

\emph{Initial clauses.} Let $J=\ell_1\vee\cdots\vee\ell_s\in F$, and let $L$ be the corresponding initial proof line
\[
L=\bigvee_{i=1}^s (1-\ell_i=0).
\]
For each $i$, the indicator of the equation $1-\ell_i=0$ is $g_{1-\ell_i}=\ell_i$, so~\eqref{eq:f-errors} gives $f\ell_iP_L\in\C_D$.
On the other hand, the falsity polynomial of $J$ is $Q_J=\prod_{i=1}^s(1-\ell_i)\in\Q(F)$. 
Using the telescoping identity $1= Q_J+ \sum_{i=1}^s \left(\prod_{j<i}(1-\ell_j)\right)\ell_i$, 
we obtain 
\begin{equation}\label{eq:syntactic-initial}
f P_L = Q_J f P_L+ \sum_{i=1}^s \left( \prod_{j<i} (1-\ell_j) \right) f \ell_i P_L.
\end{equation}
On the RHS, the displayed degree is at most $K+\Delta+s\le K+\Delta+w(F)$, where the first term follows from the input axiom $Q_J$, and every summand in the sum follows from~\eqref{eq:f-errors} by Fact~\ref{fact:multiplication}. So the LHS is in $\C_D$ by Fact~\ref{fact:multiplication}. 

For the Boolean axiom $L=(x_i=0)\vee(x_i-1=0)$, the indicators are $1-x_i$ and $x_i$, so the two error polynomials sum to $fP_L$ in degree $K+\Delta+1$, showing $fP_L\in\C_D$.

For the axiom $L=(0=0)$, its indicator is $1$, so $fP_L\in\C_D$ is already an error polynomial.
\smallskip

\emph{Weakening, contraction, and deletion of a false constant.} 
Let $L_0$ be the premise and $L$ the conclusion. Using~\eqref{eq:prefix} for $L_0$, we have
\begin{equation}\label{eq:syntactic-weakening}
 f P_L  =  P_L(fP_{L_0}) +  \sum_i U_{L_0,i} ( f g_{L_0,i} P_L ),
\end{equation}
with displayed degree $\leq K+2\Delta$ regardless of the number of disjuncts.
This shows $fP_L\in\C_D$ similarly as in \eqref{eq:syntactic-initial}, except we use the inductive hypothesis $f P_{L_0}\in\C_D$ on the first term.

\emph{A useful claim.}
We will repeatedly encounter cases where $ L_1:\ L_0 \vee (h=0)$ is a premise, and the derived line $L$ contains every equation in $L_0$. For such a case, we show 
\begin{equation}\label{eq:one-false}
 (1 - g_h)f P_L\in\C_{D}.
\end{equation}
Let $I_0$ be the index set of the equations in $L_0$. By~\eqref{eq:prefix} applied to $L_1$, then multiplied by $f P_L$,
\[
 f P_L - U_{L_1,h} g_h f P_L =  P_{L_1} f P_L + \sum_{i\in I_0} U_{L_1,i} (f g_{L_1,i} P_L)
\]
On the RHS, the first summand belongs to $\C_D$ by induction, and for every $i\in I_0$, the corresponding equation occurs in $L$, so $f g_{L_1,i} P_L \in \C_D$ by~\eqref{eq:f-errors}. 
Further multiplying this identity by $1-g_h$ and adding $U_{L_1,h} f P_L g_h(1-g_h)$ gives \eqref{eq:one-false}, with displayed degree at most $K+2\Delta+d$.
\smallskip

\emph{Linear combination.} 
Suppose 
\begin{equation}\label{eq:linear-combination}
 \frac{L_1:\ L_A\vee(h_1=0),\quad L_2:\ L_B\vee(h_2=0) }
      {L:\ L_A\vee L_B\vee(\alpha h_1+\beta h_2=0)}
\end{equation}
Denote $h_3=\alpha h_1+\beta h_2$, then 
\begin{equation}\label{eq:linear-combination-sim}
 f P_L = (1 - g_{h_1}) f P_L+ g_{h_1}(1- g_{h_2})f P_L + g_{h_1} g_{h_2}(1- g_{h_3})f P_L + g_{h_1}g_{h_2}g_{h_3} f P_L.
\end{equation}
Both $(1 - g_{h_1})f P_L$ and $(1-g_{h_2})f P_L$ are in $\C_D$ by \eqref{eq:one-false}, the third term vanishes on the Boolean cube since $h_1=h_2=0$ implies $\alpha h_1+\beta h_2=0$, and the last term is a multiple of $ f g_{h_3} P_L\in\C_D$. 
The displayed degree is at most $K+\Delta+3d\le K+2\Delta+2d$ as $d\le \Delta$.
\smallskip

\emph{Multiplication.}
Suppose the inference is from $L_1=L_0\vee(h=0)$ to $L=L_0\vee(hh'=0)$, where $\deg(hh')\le r$. Then
\begin{equation}\label{eq:multiplication-sim}
 f P_L=(1 - g_h) f P_L+ g_h ( 1- g_{hh'} ) f P_L + g_h g_{hh'} f P_L.
\end{equation}
Using \eqref{eq:one-false}, the first term is in $\C_D$; the middle summand vanishes on the Boolean cube because $h=0$ implies $hh'=0$, and the last is a multiple of $ g_{hh'} f P_L\in\C_D$. The RHS has displayed degree at most $K+\Delta+2d$.

Every step therefore fits~\eqref{eq:simdegree}. The final line has $P_\varnothing=1$, so the simulation yields $f\in\C_D$.
\end{proof}

\begin{remark}[Polynomial-clause inputs]
Proposition~\ref{prop:simulation} holds for affine systems in Remark~\ref{rmk:poly-clause} too, by replacing $w(F)$ with $W(F)$. 
Namely, only initial lines $F_i=\bigvee_{j=1}^{w_i}(H_{i,j}=0)$ need extra argument.
Since $g_{F_i,j}=\ML(1-H_{i,j}^{p-1})$ and $Q_{F_i}^{p-1}\equiv\prod_j(1-g_{F_i,j})\pmod{\Bool}$, telescoping the latter gives
\begin{equation}\label{eq:affine-remark}
fP_{F_i}\equiv Q_{F_i}^{p-1}fP_{F_i} + \sum_{j=1}^{w_i} \left(\prod_{b<j}(1-g_{F_i,b}) \right) fP_{F_i}g_{F_i,j} \mod{\Bool}.
\end{equation}
On the RHS, the first term follows from the axiom $Q_{F_i}$, and each summand follows from $fP_{F_i}g_{F_i,j}\in\C_D$ by~\eqref{eq:f-errors}.
Facts~\ref{fact:boolean} and~\ref{fact:multiplication} therefore derive $fP_{F_i}$, with the displayed degree of \eqref{eq:affine-remark} being at most $K+\Delta+W(F)$.
\end{remark}

\section{The degree--size relation}\label{sec:bounds}
We now put the parameters together. 
PC refutes $\Q(F)$ in degree $n$ for any unsatisfiable CNF, and in degree $n+1$ for multilinear polynomials as in Remark~\ref{rmk:poly-clause}.\footnote{By algebraic brute force. I.e., for each assignment $a$, choose an axiom $Q$ with $Q(a)\ne0$, multiply $Q$ successively by the literals in $a$ and multilinearizing, we derive $Q(a)\chi_a$ in degree $n+1$ where $\chi_a$ is the indicator polynomial. Scaling and summing over $a$
derives $1$.} 
So $\Q(F)\not\vdash_D1$ implies $D\le n$.

\begin{proof}[Proof of Theorem~\ref{thm:quantitative}]
Let $d=(p-1)r$. Choose
\begin{equation}\label{eq:K-t}
 K=\left\lfloor\frac D4\right\rfloor,\qquad
 t=\left\lfloor\frac{D}{8(p-1)d}\right\rfloor,\qquad
 \Delta=(p-1)dt \le \frac{D}{8}.
\end{equation}
From the assumption $D\ge\max\{32(p-1)^2r,\,8w(F)\}$, we have $K\ge 2d$. Also, since $K\le \frac n3$, $N_n(K)=\sum_{j=0}^K \binom{n}{j}\le2\binom nK$ by comparing consecutive binomials. Thus
\begin{equation}\label{eq:NnK-EST}
 \frac{N_n(K-d)}{N_n(K)}
 \ge\frac{\binom n{K-d}}{2\binom nK} = 
 \frac12\prod_{j=0}^{d-1}\frac{K-j}{n-K+j+1} \stackrel{(\text{by }K\geq 2d)}{\ge}\frac12\left(\frac{K}{2n}\right)^d \ge\frac12\left(\frac{D}{16n}\right)^d.
\end{equation}
For $\eta:=(1-\frac1p)\frac{N_n(K-d)}{N_n(K)}$, \eqref{eq:NnK-EST} gives 
\begin{equation}\label{eq:eta-EST}
 \eta t\ge\frac{D}{32pd}\left(\frac{D}{16n}\right)^d.
\end{equation}

If $\ln S<\eta t$, then $S(1-\eta)^t\le Se^{-\eta t}<1$, so by Theorem~\ref{thm:common} applied to the indicator families of the proof lines, we get the approximators $\{P_L\}_L$ and $f\notin\C_D$ such that all the $f$-multiplied error-witness polynomials belong to $\C_D$. 

Since $K+2\Delta+2d\le\frac{9D}{16}$ and $K+\Delta+w\le\frac D2$, we can apply Proposition~\ref{prop:simulation} to get $f\in\C_D$, a contradiction. So $\ln S\ge\eta t$, proving~\eqref{eq:main}.
\end{proof}

\section{Applications}\label{sec:applications}
We include a few applications of Theorem~\ref{thm:quantitative}. 
%
The following simple fact will be used.
\begin{fact}[Derivation by brute force]\label{fact:brute-force}
Let $F$ be a CNF, $C=\bigvee_{j=1}^m (h_j=0)$ where $\deg h_j\le r$, and $F$ and $C$ altogether have at most $t$ variables. If $F\models C$, then $C$ has a $\Res(\PC_r/\Fp)$ derivation from $F$ in $O(2^t t)$ size.
\end{fact}

\begin{proof}
For $\alpha\in\{0,1\}^t$, let $D_\alpha=\bigvee_{i=1}^t(x_i - (1-\alpha_i) =0 )$ be the clause falsified exactly by $\alpha$.

If $\alpha$ falsifies a clause of $F$, weakening derives $D_\alpha\vee C$. 
Otherwise, $\alpha\models F$, so some $(h=0)\in C$ satisfies $h(\alpha)=0$, and we can write $h=\sum_{i=1}^t (x_i-\alpha_i)q_i$ where $\deg q_i\le r-1$. 
By weakening the Boolean axioms we can get $D_\alpha\vee(x_i-\alpha_i=0)$ for $i\in[t]$. 
Then a multiplication by $q_i$, followed by linear combination and contraction, derives $D_\alpha\vee(h=0)$ in degree at most $r$, and weakening (where many disjuncts can be added in one step) gives $D_\alpha\vee C$. Finally, resolving the clauses $D_\alpha$ along the complete binary assignment tree derives $C$. Together this uses $O(2^t t)$ steps.
\end{proof}

\subsection{Formula families}\label{sec:app-formulas}

Throughout, we use $n$ for the number of variables, and $N$ for a natural parameter of the family.
\smallskip

\textbf{Random CNFs.}
$\mathcal F(n,w,\Delta n)$ consists of $\lceil\Delta n\rceil$ clauses chosen independently and uniformly from the $2^w\binom nw$ clauses of width $w$, with repetitions. For $\Delta\ge2^w$, the formula is unsatisfiable asymptotically almost surely by a straightforward union bound.
\smallskip

\textbf{Tseitin formulas modulo $q$.}
Fix an orientation of a connected graph $G=(V,E)$ and charges $b\colon V\to\mathbb F_q$ with $\sum_v b(v)\ne0$. 
The formula takes for each edge $e$ a set of Boolean variables $\{x_{e,a}\}_{a\in\mathbb F_q}$, and axioms include the unique-value clauses: $\vee_{a}x_{e,a}$ and $\neg x_{e,a}\vee \neg x_{e,b}$ for $a\neq b$. 
At each vertex, axioms also impose the flow-conservation property $\sum_{e\in\mathrm{out}(v)}y_e-\sum_{e\in\mathrm{in}(v)}y_e=b(v)$ where $y_e$ denotes $\sum_{a} a x_{e,a}$.
The resulting CNF $\mathrm{Tseitin}_q(G,b)$ has $q|E|$ variables and width at most the maximum of $q$ and the maximum degree of $G$.
\smallskip

\textbf{Pigeonhole principles.}
For a bipartite graph $G=(U\mathbin{\dot\cup}V,E)$, introduce a variable $X_{u,v}$ for each edge. The clauses of $\PHP(G)$ are
\begin{equation}\label{eq:PHP}
 \bigvee_{v\in N_G(u)}X_{u,v}\quad(u\in U),
 \qquad
 \neg X_{u,v}\vee\neg X_{u',v}\quad(u\ne u').
\end{equation}
The functional version $\FPHP(G)$ adds $\neg X_{u,v}\vee\neg X_{u,v'}$ for $v\ne v'$; 
the onto version $\ontoPHP(G)$ instead adds $\bigvee_{u\in N_G(v)}X_{u,v}$ for every hole $v$. 
The usual $\PHP$ formulas use the complete (bipartite) graph $G=K_{N+1,N}$. 

For $A\subseteq U$, let $\partial_G A$ be the \emph{unique neighbors set}, or \emph{boundary}, of $A$ which consists of the holes with exactly one neighbor in $A$. 
The graph is an \emph{$(s,\delta)$ boundary expander} if $|\partial_G A|\ge\delta|A|$ whenever $|A|\le s$. Here, think of $\delta>0$ as a constant. 

It is well known that such boundary expanders exist for all $N$ and $m=O(N)$. 
For instance, by a standard first-moment calculation (see, e.g.,~\cite{HLW06}), for fixed $\Delta\ge5$ a random bipartite multigraph with left degrees $\Delta$ and right degrees in $\{\lfloor\Delta m/N\rfloor,\lceil\Delta m/N\rceil\}$ satisfies that $|N(A)|\ge(\Delta-2)|A|$ for all $|A|\le\gamma N$, $\gamma>0$ a constant.
Every hole in $N(A)\setminus\partial A$ receives at least 2 edges leaving $A$, so $|\partial A|\ge2|N(A)|-\Delta|A|\ge(\Delta-4)|A|\ge|A|$. Replacing parallel edges by single edges keeps the degrees bounded and does not change neighborhoods.
\smallskip

\textbf{Colouring.}
The $3$-colouring CNF of an $N$-vertex graph $G=(V,E)$ has variables $X_{u,c}$ for $u\in V,c\in\{1,2,3\}$, and clauses requiring exactly one colour per vertex and different colours at adjacent vertices. It has $3N$ variables and width 3.
\smallskip

\textbf{Bit pigeonhole principle.}
For $N=2^\ell$ and $m>N$, $\bitPHP$ has variables $x_{i,t}$, $i\in[m]$, $t\in[\ell]$. 
Its clause-falsity polynomials are the Booleanity polynomials and 
\[\chi_{i,z}\chi_{j,z}\quad \text{for $i<j$ and $z\in\{0,1\}^\ell$},\] 
where $\chi_{i,z}=\prod_{t:\ z_t=1}x_{i,t}\prod_{t:\ z_t=0}(1-x_{i,t})$. We call the corresponding clauses the collision clauses.\smallskip

\begin{theorem}[Exponential lower bounds]\label{thm:applications}
Fix a prime $p$ and an integer $r\ge1$. Each of the following requires proof size $\exp(\Omega_{p,r}(N))$ in $\Res(\PC_r/\Fp)$.

\begin{enumerate}[label=(\alph*)]
\item 
$F\sim\mathcal F(N,w,\Delta N)$, for fixed $w\ge3$ and fixed $\Delta\ge 2^w$, asymptotically almost surely. \label{item:random-cnf} 

\item 
$\mathrm{Tseitin}_q(G,b)$, for a fixed prime $q\ne p$ and a family of fixed-degree regular edge expanders on $N$ vertices. \label{item:tseitin} 

\item
The $3$-colouring CNF of $G(N,\lambda/N)$ or of a uniform $\lambda$-regular graph on $N$ vertices, for fixed $\lambda\ge6$, asymptotically almost surely. \label{item:colouring} 

\item
$\mathrm{PHP}(G)$, $\FPHP(G)$, and $\ontoPHP(G)$, with $N+1$ pigeons and $N$ holes: both the complete graphs and bounded-maximum-degree $(s,\delta)$ boundary expanders with $s=\Omega(N)$, fixed $\delta\ge1$, and no isolated holes. \label{item:php} 

\item
Bit-PHP with $m=O(N)$ pigeons and $N=2^\ell$ holes. \label{item:bit-php}
\end{enumerate}
\end{theorem}

\begin{proof}
For (a), (b), (c), and the bounded-degree graphs in (d), we use linear PC degree lower bounds on $n=\Theta(N)$ variables from the literature and apply Theorem~\ref{thm:quantitative}. The usual $\PHP$ and $\bitPHP$ follow by reductions.
\smallskip

\emph{Random CNFs.}
For fixed $w\ge3$ and fixed $\Delta$, Alekhnovich and Razborov~\cite{AR03} prove that every PC refutation of $F\sim\mathcal F(n,w,\Delta n)$ has degree at least $\varepsilon n$ asymptotically almost surely, where $\varepsilon=\varepsilon(w,\Delta)>0$. 
(See also~\cite{BI10} over characteristic greater than~$2$).
\smallskip

\emph{Tseitin.}
Linear PC degree lower bounds over characteristic different from $q$ for the mod-$q$ Tseitin are proved in \cite{AR03,BGIP01}. The encodings are different from here, but since $q$ and the graph degree are fixed, we can substitute $x_{e,a}$ by the variable that indicates edge $e$ takes value $a$ in their encoding, which will increase degree by at most $q-1$. 
The substitution maps every axiom of $\mathrm{Tseitin}_q(G,b)$ to a constant-degree consequence of the axioms there, so a degree-$D$ PC refutation here yields a degree-$O(D)$ refutation in those encodings. Thus the linear lower bound transfers. 
\smallskip

\emph{Pigeonhole.}
Mik\v{s}a--Nordstr\"om~\cite{MN} prove degree lower bounds of $\delta s/(2b)$ for graph FPHP$(G)$ on such graphs of maximum left degree $b$, and of $\delta s/2$ for $\ontoPHP(G)$. The formula $\PHP(G)$ is a weakening of $\FPHP(G)$, so size lower bounds transfer. 

For the complete graph, choose such a bounded-degree graph $G$. The graph version formula is a restriction of the complete-graph case by restricting all non-edge variables to $0$, which only reduces proof degree and size. So the lower bounds transfer to the complete-graph case.
\smallskip

\emph{Colouring.}
Conneryd, de Rezende, Nordstr\"om, Pang, and Risse~\cite{CdRNPR} prove linear degree for the encoding with polynomials $R_u=\sum_cX_{u,c}-1$, $X_{u,c}X_{u,c'}$ ($c\neq c'$), $X_{u,c}X_{v,c}$ ($uv\in E$), and Booleanity. 
This encoding derives the falsity polynomials of the CNF axioms in degree 3, using $\prod_{c=1}^3(1-X_{u,c})=-R_u+\sum_{c<c'}X_{u,c}X_{u,c'}-X_{u,1}X_{u,2}X_{u,3}$.
Thus the same degree lower bound holds for the CNF.
\smallskip

\emph{Bit pigeonhole.}
Choose a boundary expander $G$ as above with $m$ left vertices (pigeons), $N$ right vertices (holes), bounded left degree $b$, and expansion up to $\Omega(N)$ pigeons; the $\FPHP$ degree bound above is $\Omega(N)$ on $\Theta(N)$ variables. 
Label the holes by $\{0,1\}^\ell$, then the affine substitution
\begin{equation}\label{eq:app-bit-decoder}
 \sigma(x_{i,t}) = \sum_{\substack{v\in N_G(i)\\v_t=1}}X_{i,v}
\end{equation}
turns a refutation of $\bitPHP$ into one of $\FPHP(G)$. It preserves the degree bound of each line, and it preserves deduction steps without increasing proof size.\footnote{Recall that we measure proof size by the number of lines, regardless of the size of each polynomial disjunct.} 
For axioms, the image of a collision clause is semantically implied (over the Boolean cube) from the $\FPHP(G)$ constraints on the corresponding two pigeons, and the image of a Boolean axiom from the $\FPHP(G)$ constraints on that pigeon. 
Each involves at most $2b$ variables, so Fact~\ref{fact:brute-force} derives it in $O_b(1)$ steps in degree 1. Thus a size-$S$ $\bitPHP$ refutation yields a size-$O_b(S)$ refutation of $\mathrm{FPHP}(G)$. 
The $\FPHP$ argument of \ref{item:php}, which uses only $m=O(N)$, completes the proof.
\end{proof}

For $\Res(\oplus_2)$, item \ref{item:bit-php} improves the $\exp(\Omega(N/\log^2N))$ lower bound in~\cite{Braun2026exponential} to $\exp(\Omega(N))$. 
This is nearly tight since tree-like resolution can refute $\bitPHP$ in size $\exp\left(O(N)\right)$~\cite{DGGM24}.

\subsection{Different moduli}
\label{sec:app-lin}

Taking $r=1$ gives the same bounds for $\Res(\mathrm{lin}_{\Fp})$, so Family (b) also separates the characteristics. Previously, Part--Tzameret~\cite{PartTzameret} showed the incomparability for tree-like refutations.

\begin{corollary}\label{cor:moduli}
Let $p\ne q$ be primes and let $G$ belong to the graph family in Theorem~\ref{thm:applications}(b). Then $\mathrm{Tseitin}_q(G,b)$ has size-$O(N)$ $\Res(\mathrm{lin}_{\mathbb F_q})$ refutations. 
In particular, $\Res(\mathrm{lin}_{\Fp})$ and $\Res(\mathrm{lin}_{\mathbb F_q})$ are exponentially incomparable.
\end{corollary}
\begin{proof}
Given the lower bound in Theorem~\ref{thm:applications}(b), we only need to derive the upper bound. Using Fact~\ref{fact:brute-force}, we can use the CNF to derive the flow-conservation equation at each vertex, and their sum over all vertices gives the contradiction $-\sum_vb(v)=0$. 
\end{proof}

\subsection{Resolution over \texorpdfstring{$k$}{k}-DNFs}\label{sec:app-resk}

The lines of $\Res(k)$~\cite{Krajicek01} are disjunctions of Boolean terms of at most $k$ literals. We use the standard syntactic rules
of weakening,
\[
 \frac{A\vee T_1\qquad B\vee T_2}
      {A\vee B\vee(T_1\wedge T_2)}
 \quad (|T_1\cup T_2|\le k)
\]
for $\wedge$-introduction,
$A\vee T\mapsto A\vee T'$ for $T'\subseteq T$ for $\wedge$-elimination, and
\[
 \frac{A\vee T\qquad B\vee\bigvee_{\ell\in T}\neg\ell}
      {A\vee B}
\]
for cut. Terms and disjunctions are written without repetitions.

\begin{fact}\label{fact:resk}
For every CNF $F$, an $S$-node $\Res(k)$ refutation yields a size-$O(kS)$ $\Res(\PC_k/\Fp)$ refutation, for every prime $p$.
\end{fact}

\begin{proof}
We can translate a $\Res(k)$ line into one in $\Res(\PC_k/\Fp)$, by representing a Boolean term $T$ by equation $( m_T -1 =0)$, where $m_T$ is the product of its literal truth polynomials. We simulate deductions below. 

Weakening and contraction are immediate.

For other rules, we will use the fact that $B_m: (m=0)\vee(m-1=0)$ has a size-$O(\deg m)$ derivation in degree $\deg m$ from the Boolean axioms (cf. \cite[Proposition~4.1]{Khaniki}).

For $\wedge$-elimination $T\mapsto T'\subseteq T$, derive $B_{m_{T'}}$ and multiply its equation $(m_{T'}=0)$ by $m_T/m_{T'}$. 
Then combining the resulting $(m_T=0)$ with the premise equation $(m_T-1=0)$ gives a false constant, whose deletion leaves $(m_{T'}-1=0)$.

For $\wedge$-introduction of $T_1,T_2$, we can first eliminate from $T_2$ the literals already in $T_1$ using $\wedge$-elimination, and then derive $m_{T_2\setminus T_1}(m_{T_1}-1) + (m_{T_2\setminus T_1}-1) = m_{T_1\cup T_2}-1$.

For a cut on $T=\ell_1\wedge\cdots\wedge\ell_j$ ($j\leq k$), the translation of $\bigvee_i \neg \ell_i$ contains the disjunct $(\ell_i=0)$ for every $i$. 
We multiply each by $m_T/\ell_i$, contract to obtain $(m_T=0)$, and then combine this with $(m_T-1=0)$ from $T$ to get a false constant, which is then deleted. 

All displayed polynomials have degree at most $k$, and each $\Res(k)$ inference is simulated using $O(k)$ steps. So the whole refutation has size $O(kS)$.
\end{proof}

\begin{corollary}\label{cor:resk}
Let $k=k(n)$ be a parameter and let $F$ be a CNF on $n$ variables. 
If $\Q(F)$ has no PC refutation over $\mathbb F_2$ in degree $D$, where $D\ge\max\{32k,8w(F)\}$, then every $\Res(k)$ refutation of $F$ has DAG-size 
\begin{equation}\label{eq:app-resk}
 S\ge\frac{1}{O(k)}\exp\left[\frac{D}{128k}\left(\frac{D}{16n}\right)^k\right].
\end{equation}
In particular, when $D=\Omega(n)$, $S\ge\exp\left( n/2^{O(k)} \right)$.
\end{corollary}
\begin{proof}
Apply Fact~\ref{fact:resk} and Theorem~\ref{thm:quantitative} with $p=2,r=r(n)=k$. 
\end{proof}

For every family in Theorem~\ref{thm:applications}, we thus obtain for every $k\ge 1$
\[
 S\ge \exp\left(N/2^{O(k)}\right).
\]
Here, for Tseitin modulo $2$, we apply Theorem~\ref{thm:quantitative} and Fact~\ref{fact:resk} with $p=3$ instead. 

For random constant-width CNFs, including 3-CNFs, this extends the $k=O(\sqrt{\log n})$ range in prior works \cite{Carenini26,SofronovaSokolov25}.
For $\bitPHP$, earlier near-exponential bounds hold for $k=O((\log N)^{1/2-\varepsilon})$ for arbitrary $m$~\cite{DGGM24}, while our stronger bound applies to $m=O(N)$ covered by the boundary expanders used here.
For modulo-two Tseitin and $\PHP$, the $\Res(k)$ lower bounds for logarithmic $k$ also follow from the known bounded-depth Frege lower bounds \cite{BenSasson02,BIKPPW92}.
For $\PHP$, Razborov~\cite{Razborov15} proved $\exp(N/(\log N)^{O(k)})$, and here the denominator improves to $2^{O(k)}$.

\subsection{Proof search}\label{sec:app-search}

Since degree-bounded PC refutations can be efficiently found \cite{CEI96}, Theorem~\ref{thm:quantitative} gives the following proof search algorithm as a direct analogue of those in~\cite{B-SW,IPS}.

\begin{corollary}\label{cor:search}
Fix $p,r$. Let $d=(p-1)r$. 
If a width-$w$ CNF $F$ on $n$ variables has a $\Res(\PC_r/\Fp)$ refutation of size $S\geq 2$, then a PC refutation of $\Q(F)$ can be found in time
\begin{equation}\label{eq:search}
 |F|^{O(1)}\cdot  n^{O\big( w+ n^{d/(d+1)} (\ln S)^{1/(d+1)} \big)}.
\end{equation}
The algorithm need not know $S$.
\end{corollary}

\begin{proof}
Theorem~\ref{thm:quantitative} implies the existence of a PC refutation of $\Q(F)$ in degree
\[
 D=O_{p,r}\left(  w+ n^{d/(d+1)} (\ln S)^{1/(d+1)} \right).
\]
The search for a PC refutation in increasing degree $D$ will succeed in $|F|^{O(1)}n^{O(D)}$ time.
\end{proof}

For $\Res(\oplus_2)$, the runtime in \eqref{eq:search} is $|F|^{O(1)}n^{O(w+\sqrt{n\ln(2S)})}$. 
For $\Res(k)$ with $k=k(n)$, by Fact~\ref{fact:resk} and Theorem~\ref{thm:quantitative} taking $p=2,r=k$, we get the analogue of~\eqref{eq:search} with exponent $O\bigl(w+k+n^{k/(k+1)}(k\ln(kS))^{1/(k+1)}\bigr)$ over $n$.

\section{Conclusions}\label{sec:conclusions} 
The rank-based proof approximation behind Theorem~\ref{thm:quantitative} extends to $\ACzerop$-Frege, basically because rank is subadditive. We make this explicit below.

The Frege system can be formulated in many different but polynomially equivalent ways.
For definiteness, we fix one. 
A proof line is an $\ACzerop$-formula, i.e., a rooted tree whose leaves are Boolean variables and non-leaf nodes are \emph{gates} of types $\neg,\vee,\MOD_{p,a}\ (a\in\Fp)$, the latter two allowing unbounded fan-in. 
We formally represent conjunction by the De Morgan dual of disjunction. 
Formula \emph{depth} is the maximum number of $\vee,\MOD_{p,a}$ gates on a root-to-leaf path (negations do not count). 
A \emph{subformula} is a subtree identified by a gate.

\begin{definition}[{$\ACzerop$-Frege}]\label{def:ac0p-frege}
The unbounded-disjunction Frege system~\cite{BPU92} has lines as $\ACzerop$-formulas. 
The deduction rules include \emph{weakening} and \emph{cut},
\begin{equation}\label{eq:wk-ct}
\frac{A}{A\vee B},\qquad
\frac{A\vee B\qquad\neg A\vee C}{B\vee C}.
\end{equation}
Disjunctions are indexed by sets, so the \emph{permutation} and \emph{contraction} rules are automatic. 
It also has the \emph{merging}/\emph{unmerging} rule (i.e., unnesting/nesting ORs):
\begin{equation}\label{eq:merge}
\frac{\bigvee(\{\bigvee\Gamma\}\cup\Delta)}
     {\bigvee(\Gamma\cup\Delta)},
\qquad
\frac{\bigvee(\Gamma\cup\Delta)}
     {\bigvee(\{\bigvee\Gamma\}\cup\Delta)}
\end{equation}
where $\Gamma,\Delta$ are sets of formulas. 
The gate $\MOD_{p,a}(\bar A)=\MOD_{p,a}(A_1,\ldots,A_k)$ represents that the Boolean truths of $A_1,\ldots,A_k$ sum to $a$ mod $p$, with axiom schema~\cite[Definition~1.1]{BIKPRS96}:
\[
\begin{aligned}
\mathrm{(M1)}\quad& \MOD_{p,0}(),\\
\mathrm{(M2)}\quad& \neg \MOD_{p,a}() \qquad (1\le a<p),\\
\mathrm{(M3)}\quad& \MOD_{p,a}(\bar A,B)\leftrightarrow \bigl[(\MOD_{p,a}(\bar A)\wedge\neg B)\vee(\MOD_{p,a-1}(\bar A)\wedge B)\bigr].
\end{aligned}
\]
Here, $(A_1,\ldots,A_k)$ is treated as a multi-set so permutation of arguments is automatically allowed. 
The $\leftrightarrow$ sign in (M3) is encoded using $(\neg X\lor Y)\bigwedge(X\lor \neg Y)$ for $X\leftrightarrow Y$. 
The \emph{excluded-middle schema} $A\vee \neg A$ is also allowed. 
The depth-$\ell$ system allows only formulas of depth at most $\ell$. 
Proof \emph{size} counts the number of gates, including leaves, across proof lines.
\end{definition}

For very small depth $\ell\in\{2,3,4\}$, we encode (M3) using four clauses: with $X=\MOD_{p,a}(\bar A)$, $Y=\MOD_{p,a}(\bar A,B)$, $Z=\MOD_{p,a-1}(\bar A)$, they are
\begin{equation}\label{eq:M3-alternative}
B\vee X \vee \neg Y, \ \ B\vee \neg X\vee Y,\ \ \neg B\vee\neg Y\vee Z, \quad \neg B\vee Y\vee\neg Z.
\end{equation} 

The following holds for all $\ell\geq 2$. 
\begin{restatable}[Conditional {$\ACzerop$}-Frege lower bounds]{proposition}{frege}\label{prop:frege}
Fix a prime $p$ and an integer $\ell\ge2$. 
Let $(F_n)_n$ be a CNF family on $n$ variables where $\Q(F_n)\not\vdash_D 1$, for some $D= w(F_n) + \Omega(n)$. 
Then every depth-$\ell$ $\ACzerop$-Frege refutation of $F_n$ has size at least
\begin{equation}\label{eq:AC0p-bound}
 \exp\!\left[
 	\Omega\left(\min\left\{ n^{\frac{0.99}{2\ell-1}}, \quad 
 	\left(\frac nB\right)^{\!\frac1{\ell-2}},\quad 
 	\left(\frac{n\log n}{B}\right)^{\!\frac1{\ell-1}}\right\}\right)
\right],\ B:=\max\Big\{0,\frac n2 - \big(D-w(F_n)\big)\Big\}.
\end{equation}
Here, the terms involving $\frac1B$ are omitted when $B= 0$, and the middle term is omitted if $\ell=2$. 

In particular, if
\begin{equation}\label{eq:frege-threshold}
 D-w(F_n)\ge\frac n2-o\!\left(\frac n{(\log n)^{\ell-2}}\right),
\end{equation}
then every depth-$\ell$ $\ACzerop$-Frege refutation has size $n^{\omega(1)}$.
\end{restatable}

The $0.99$ in the first term of \eqref{eq:AC0p-bound} can be replaced by any constant below 1. 
The second and third terms of \eqref{eq:AC0p-bound} penalize $D - w(F)$ when it is far below $\frac n2$; they come from the approximation constraints on gates at depth $<\ell$ and at depth $\ell$, respectively. 

\smallskip

The proof is not surprising after Theorem~\ref{thm:quantitative}. We include it in Appendix~\ref{sec:appendix} for two purposes: to clarify parameter dependence, and to streamline the simulation.

\paragraph{Open problems.} 
The main question is to prove a PC degree lower bound near $n/2$ as needed above, or eliminate errors with a better rate $\eta_K(\delta)$ such as by constructing tailored error measures for special formulas.

Over characteristic zero, a general degree--size relation has to fail: graph $\FPHP$ on bounded-degree boundary expanders requires linear PC degree, despite having polynomial-size $\Res(\mathrm{lin}_{\mathbb Q})$ refutations~(cf. \cite{RazTzameret08}).
Can one nevertheless prove exponential lower bounds for constant-width CNF (cf. \cite{PartTzameret,Khaniki,Alekseev26})? 

It would also be interesting to either improve the exponent in~\eqref{eq:main} or find formulas witnessing its tightness, in light of the results in \cite{BonetGalesi01,GalesiLauria10} for resolution and PC/PCR.

\ifthenelse{\boolean{personalsite}}{}{
\paragraph{AI Disclosure.}
The author studied \cite{Braun2026exponential} by asking targeted questions with ChatGPT-6 Extra High. He then supplied a simplified proof architecture to ChatGPT-6 Pro, together with a suggestion to use a multiplier to restore inferences disrupted by Razborov--Smolensky.

The initial exchanges identified obstructions to certain approaches for characteristics other than $2$. After several more, the model proposed a lower-bound proof for $\bitPHP$ in $\Res(\oplus_p)$. The author subsequently suggested seeking a measure of the errors relative to the input formula. The model developed this into a rank measure modulo bounded-degree PC, and proposed essentially the arguments in Section~\ref{sec:rank}. 

Clearly the suggestions directed the search materially; only their necessity remains unclear. 
The author studied, anatomized and reconstructed the mathematics, then wrote the paper with AI assistance at that stage limited to editorial suggestions.

\paragraph{Acknowledgements} 
I thank Jonas Conneryd and Yassine Ghannane who brought \cite{Braun2026exponential} to my attention, and Jonas Conneryd for suggesting the question in characteristic 0.
}

{\small 
\setstretch{0.9}

\bibliographystyle{alpha}
\bibliography{refs}  

}

\appendix
\section{Proof of Proposition~\ref{prop:frege}}\label{sec:appendix}

We use an elementary fact.
\begin{fact}\label{fact:binomial} 
Let $a(n),b(n)\ge 0$ be integer parameters with $\frac{a+b}n \leq \frac12 - \varepsilon$ for some constant $\varepsilon>0$. Write $K= \lfloor\frac{n}{2} - a\rfloor$, then 
\begin{equation}\label{eq:frege-est}
 \frac{N_n(K-b)}{N_n(K)}\ge \exp\!\left[- O_\varepsilon\left(\frac{b(a+b+\sqrt n)}n\right)\right].
\end{equation}
\end{fact}
\begin{proof}
For any $k \in [\varepsilon n,\ \frac{n}{2}]$, we have
\begin{equation}\label{eq:EST-1}
\frac{\binom n k}{N_n(k)} \le O_\varepsilon\left(\frac{(\frac n2 - k) +\sqrt n}{n}\right).
\end{equation}
Indeed, the ratio {$\dfrac{\binom ni}{\binom n{i+1}}$} is increasing in $i$, so for $0\le i\le k$, 
{$\dfrac{\binom n{k-i}}{\binom nk} \ge \left(\dfrac{\binom n{k-i}}{\binom{n}{k-i+1}}\right)^i = 
\left(\dfrac{k-i+1}{n-k+i}\right)^i$}. 
The expression {$\dfrac{k-i+1}{n-k+i}$}, by anchoring its numerator and denominator at $\dfrac n2$, is equal to
${\dfrac{1 - x +\frac2n}{1 +  x}}$ for ${x=\dfrac{(n/2 - k)+i}{n/2}}$, so it is lower bounded by 
$\exp\left[-O_\varepsilon(x)\right]$ when $x \leq 1-\varepsilon$, which holds when $i\leq \dfrac k2$.
In particular, for {$0\leq i \leq \dfrac{n}{(n/2-k)+\sqrt n} \ll \dfrac k2$}, 
\[\frac{\binom n{k - i}}{\binom n k} \ge \exp\left[ - O_\varepsilon\left(\frac{i( (n/2 - k) +i)}{n}\right) \right]=\Omega_\varepsilon(1).\]
Summing over these $i$'s then taking reciprocal gives \eqref{eq:EST-1}.

We also have that {$\dfrac{\binom{n}{k}}{N_n(k)}$} is bounded away from 1: 
{$\dfrac{\binom{n}{k}}{N_n(k)} \le \dfrac{\binom nk} {\binom nk+\binom n{k-1}} = \dfrac{n+1-k}{n+1}\le 1 - \dfrac{\varepsilon}{2}$}. 
Thus $1-\frac{\binom{n}{k}}{N_n(k)}\geq \exp\left( - O_\varepsilon(\frac{\binom{n}{k}}{N_n(k)})\right)$. Together, 
\begin{equation}\label{eq:EST-2}
\frac{N_n(k-1)}{N_n(k)} =1-\frac{\binom nk}{N_n(k)}\ge \exp\left[ - O_\varepsilon\left(\frac{\binom{n}{k}}{N_n(k)}\right)\right]
\stackrel{\eqref{eq:EST-1}}\ge \exp\left[ - O_\varepsilon\left(\frac{(n/2-k)+\sqrt n}{n}\right)\right].
\end{equation}
Apply \eqref{eq:EST-2} to $k=K,K-1,\ldots,K-b+1$ which are all at least $\varepsilon n$, and multiply the result, 
\[
\frac{N_n(K-b)}{N_n(K)} \ge \exp\left(-C_\varepsilon \sum_{j=0}^{b-1}\frac{a+j+\sqrt n}{n}\right) \ge 
\exp\left(-C_\varepsilon \frac{b(a+b+\sqrt n)}n\right). \qedhere
\]
\end{proof} 

\frege*

\begin{proof}[Proof of Proposition~\ref{prop:frege}]
Take $F=F_n$ and $D$ as in the proposition. 
Without loss of generality, assume $D-w(F)\le \frac n2$; otherwise we can take a smaller $D$. 

Let $\Pi$ be a refutation of $F$. Since $D=\Omega(n)$ and PC simulates $\Pi$ in degree $O_{p,\ell}(|\Pi|)$, we have $|\Pi|\to\infty$. Below, denote
\begin{equation}\label{eq:frege-L}
L:=\ln(2|\Pi|),\ \ \text{so $L\to\infty$.}
\end{equation} 
Let $\mathcal S_\Pi$ be the set of all distinct subformulas in $\Pi$, then $|\mathcal S_\Pi| \le|\Pi|$. 
The proof plan is to first construct approximators for all formulas in ${\mathcal S_\Pi}$ by induction on depth, choose a multiplier, then simulate the $f$-ORified version of $\Pi$. 

For $B=\lfloor\frac{n}{2}-(D-w(F))\rfloor$, we show that if
\begin{equation}\label{eq:frege-trial}
 L \leq c\cdot \min\left\{ n^{\frac{0.99}{2\ell-1}},\ \left(\frac nB\right)^{\!\frac1{\ell-2}},\ \left(\frac{n\log n}{B}\right)^{\!\frac1{\ell-1}}
 \right\}
\end{equation}
for a sufficiently small constant $c$ (where the $B=0$ and $\ell=2$ conventions carry over), then this simulation is in degree $D$ and this would reach a contradiction. 
\smallskip 

\textbf{1. Approximator construction.} 
For each formula $A\in{\mathcal S_\Pi}$, we construct a Boolean-valued polynomial $P_A$ (the ``approximator''), where $P_A=0$ approximates the predicate $A$. We induct on formula depth and on formula structure within the same depth.
At depth $0$, we assign literals their falsity polynomials of degree $\Delta_0=1$. 

Suppose $A$ is a formula (as a gate) of depth $i$. If it is an OR gate $A=\lor_{j=1}^s A_j$, then by the inductive hypothesis we have constructed approximators of its children with degrees at most $\Delta_{i-1}$. 
Define $P_A$ by~\eqref{eq:product} applied to $1-P_{A_1},\ldots,1-P_{A_s}$ using $t_i$ factors: 
\begin{equation}\label{eq:PA-appendix}
P_A = \prod_{j=1}^{t_i} ( 1 - q_{A,j}^{p-1} ),\quad\text{where}\ \ q_{A,j}\in\Span_{\Fp}\{1-P_{A_1},\ldots,1-P_{A_s}\}.
\end{equation}
The error polynomials are thus
\begin{equation}\label{eq:eAj}
 e_{A,j}=P_A(1-P_{A_j}), \quad j\in[s].
\end{equation}
Let $T_A$ be the stacked map $T_A\colon\V_K\longrightarrow\bigoplus_{\alpha\in[s]}(\V_D/\C_D)$ defined by $v\longmapsto ([ve_{A,1}],\ldots,[ve_{A,s}])$, then by Theorem~\ref{thm:common}, we can find such a $P_A$ satisfying 
\begin{equation}\label{eq:P-A-T-A}
 \deg P_A\le\Delta_i=(p-1)\Delta_{i-1}t_i
 \quad\text{and}\quad
 \rank T_A\le h_K(1-\eta_K(\Delta_{i-1}))^{t_i},
\end{equation}
provided $\Delta_{i-1}\le K$ and $K+\Delta_i+\Delta_{i-1}\le D$. Here, recall 
\begin{equation}\label{eq:eta-x}
 h_K=\dim\ \V_K/(\C_{D}\cap \mathcal V_K),\quad \eta_K(x)=\left(1-\frac1p\right)
 \frac{N_n(K-x)}{N_n(K)}.
\end{equation}
If $A$ is a modular or negation gate, we let
\begin{align}
P_{\MOD_{p,a}(A_1,\ldots,A_k)} &=\left(\sum_{j=1}^k(1-P_{A_j})-a\right)^{p-1}\label{eq:MOD-gate}\\
P_{\neg A} &=1-P_A\label{eq:neg-gate}
\end{align} 
These are exact given the child approximators, introducing no new error polynomials. In the modular gate case, the degree is at most $(p-1)\Delta_{i-1}$.

\textbf{2. Parameter choice.}
We choose 
\begin{align}
 K &= \lfloor D-w(F) - \frac{n}{L^{\ell-1}}\rfloor \label{eq:choice-K}\\
 t_i &=\left\lceil\frac{L}{\eta_K(\Delta_{i-1})}\right\rceil\qquad\text{$\forall i\in[\ell]$}.\label{eq:choice-t-i}
\end{align}
By~\eqref{eq:P-A-T-A}, 
\begin{equation}\label{eq:frege-conditions}
\rank T_A\le h_K/(2|\Pi|)\quad \text{conditioned on}\quad \Delta_\ell \le K \text{ and } K+2\Delta_\ell \le D.
\end{equation}
We examine these conditions below.  
\smallskip 

\emph{At depth $i<\ell$.} We maintain $\Delta_i=O_{p,\ell}(L^i)$ inductively over $i$. At step $i$, we apply Fact~\ref{fact:binomial} using $a= \lfloor \frac n2-K\rfloor$, $b= \Delta_{i-1}$. 
Since $\frac{\Delta_{i-1}\Big( (\frac n2 - K) +\Delta_{i-1}+\sqrt n\Big)}n = O_{p,\ell}\left(\frac 1L + \frac{B L^{\ell-2} + L^{2(\ell-2)}}n + \frac{L^{\ell - 2}}{\sqrt n}\right)$ where $\lfloor\frac n2 - K \rfloor = B + \lceil\frac{n}{L^{\ell-1}}\rceil$, it suffices to impose:
\begin{align}
\frac{n}{L^{\ell -1}}+ L^{\ell -1} = & \ o\big(n\big), \tag{Constraint 1}\label{eq:constraint-1}\\
\frac{B L^{\ell-2} + L^{2(\ell-2)}}n + \frac{L^{\ell - 2}}{\sqrt n} = & \ O(1). \hspace{30pt} \tag{Constraint 2}\label{eq:constraint-2}
\end{align}
Indeed, if so, then $\Delta_{i-1} = O_{p,\ell}(L^{i-1})$ and the condition in Fact~\ref{fact:binomial} holds, whose application gives $\eta_K(\Delta_{i-1})=\Omega_{p,\ell}(1)$, thus $t_i=O_{p,\ell}(L)$ by \eqref{eq:choice-t-i}, and so $\Delta_i=O_{p,\ell}(L^i)$, closing the induction. 

\emph{At depth $\ell$.} Fact~\ref{fact:binomial} gives $\eta_K(\Delta_{\ell-1}) = \exp\!\left[- O_{p,\ell}\!\left(1+\frac{B L^{\ell-1} + L^{2(\ell-1)}}n + \frac{L^{\ell - 1}}{\sqrt n}\right)\right]$, so \eqref{eq:choice-t-i} gives
\begin{equation}\label{eq:Delta-ell}
 \Delta_\ell =  (p-1)\Delta_{\ell-1}t_\ell \le L^\ell \exp\!\left( O_{p,\ell}\!\left(1+\frac{B L^{\ell-1} + L^{2(\ell-1)}}n +\frac{L^{\ell - 1}}{\sqrt n}\right) \right).
\end{equation}
These approximators are not fed into other OR gates, so we don't need to inspect $\eta_K(\Delta_\ell)$. 
We only need to ensure the two conditions in \eqref{eq:frege-conditions}; for the purpose of simulation later, we strengthen the second condition to  
\begin{equation}\label{eq:sim-reserve}
K+w(F)+O_{p,\ell}(\Delta_\ell)\le D.
\end{equation} 
So we impose $\Delta_\ell = o\left( \frac n{L^{\ell -1}} \right)$, or
\begin{align}
\hspace{16pt}L^\ell \exp\left[ O_{p,\ell}\left(\frac{BL^{\ell -1} + L^{2(\ell - 1)}}{n}\right)\right] = o\left(\frac{n}{L^{\ell - 1}}\right). \tag{Constraint 3}\label{eq:constraint-3}
\end{align}
Indeed, the first condition in $\eqref{eq:frege-conditions}$ also follows from \ref{eq:constraint-1}, \ref{eq:constraint-3}, and the assumption that $D-w(F)=\Omega(n)$: $K-\Delta_\ell = (D-w(F)) - \lceil\frac{n}{L^{\ell-1}}\rceil - \Delta_\ell = \Omega(n) - o(n) \ge 0$. 

Constraints 1, 2 and 3 are all we need. Since $L\to\infty$, they hold if
\[
\frac{BL^{\ell-2}}n=O(1) \ \ \text{and} \ \ \frac {L^{2\ell-1}}n \exp\left(C\frac{BL^{\ell-1}}n\right)=o(1),\ \ \text{where $C=C_{p,\ell}$ is pre-given but constant.}
\]
Condition \eqref{eq:frege-trial} is one way to guarantee both. 

Therefore, assuming \eqref{eq:frege-trial}, the conditions in \eqref{eq:frege-conditions} hold throughout. Then the sum of $\rank T_A$ over distinct OR gates in $\Pi$ is at most $h_K/2<h_K$, so there exists $f\in\V_K\setminus\C_D$ such that
\begin{equation}\label{eq:or-polys}
 fP_A(1-P_{A_j})\in\C_D
 \qquad\text{for every OR gate $A=\bigvee_j A_j$ and every $j$.}
\end{equation} 
This completes the construction of the approximators $P_A$ and multiplier $f$. 

We record an identity for OR gates by a telescoping of the definition \eqref{eq:PA-appendix}, similar to \eqref{eq:prefix}:
\begin{equation}\label{eq:or-telescoping}
 1-P_A=\sum_j U_{A,j}(1-P_{A_j}),\quad\text{for some $U_{A,j}$ of degree} \le\Delta_\ell.
\end{equation}

\textbf{3. Syntactic simulation.}
We give a uniform treatment of the simulation of Frege axioms and rules via the following lemma; cf. \cite[Section 6]{BIKPRS96} for a similar underlying principle.

The idea is to treat the leaves of an axiom schema as formal variables, then simulate any predicate over the occurring gates. 

Formally, let a \emph{template} $T$ be a finite forest of formulas whose leaves are formal names. 
We associate a variable $z_v$ to each node, with repeated leaf names sharing a variable and all other node variables distinct. 
The equation $z_v=0$ is intended to represent the node value being TRUE. 
For each OR gate $v=\bigvee_i v_i$, we also introduce fresh variables $u_{v,i}$, which needn't be Boolean. 
Let $\mathsf G(T)$ consist of the Booleanity polynomials $z_v^2-z_v$ ($\forall v\in T$) and the following gate polynomials:
\begin{align}
 z_v+z_w-1 &\qquad(v=\neg w) \label{eq:T-neg}\\
 z_v-\left(\sum_i(1-z_{v_i})-a\right)^{p-1} &\qquad(v=\operatorname{MOD}_{p,a}(v_1,\ldots,v_s))\label{eq:T-mod}\qquad\\
 z_v(1-z_{v_i}) &\qquad(v=\bigvee_{i=1}^s v_i,\ \ i\in[s]) \label{eq:T-or-err}\\
 1-z_v-\sum_i u_{v,i}(1-z_{v_i}) &\qquad(v=\bigvee_{i=1}^s v_i) \label{eq:T-or-telescoping}
\end{align}

\begin{lemma}[Gate semantics]\label{lem:gate-sim}
Let template $T$ contain $t_O$ OR gates and $t_M$ modular gates, with arbitrary fan-in. If $R\in\Fp[(z_v)_{v\in T}]$ vanishes on every Boolean assignment that is consistent with the gate computations in $T$, then $\mathsf G(T)$ derives $R$ in PC degree
\begin{equation}\label{eq:T-deg}
  \deg R+3t_O+(p-1)^2t_M+O(1),
\end{equation}
regardless of the number of leaves and negations.
\end{lemma}

\begin{proof}
Think of the $z$-variables as arranged on $T$, then recursively bottom-up, we clear all appearing negation-gate variables by substituting $z_{\neg v}=1-z_v$. 
So a chain of negations is affine in its first non-negation descendant. Using \eqref{eq:T-neg}, PC can derive from $\mathsf G(T)$ its substituted version, in degree $O_p(1)$. 
Below, by $\mathsf G(T)$, we refer to this substituted version, and the same for $R$.

The plan is to derive $R$ conditioned on each complete assignment $\alpha$ to the OR gates, then add them. 

For each assignment $\alpha\in\{0,1\}^{t_O}$ to the OR gates in $T$, let 
\begin{equation}\label{eq:chi-alpha}
\chi_\alpha= \prod_{\alpha_v=1}z_v \prod_{\alpha_v=0}(1-z_v).
\end{equation}
Then $\sum_\alpha\chi_\alpha=1$. For each $\alpha$, let $I_\alpha$ be the set of OR gates $v$ with $\alpha_v=0$. 
For $v\in I_\alpha$, from \eqref{eq:T-or-telescoping} and Booleanity, we have
\begin{equation}\label{eq:T-telescoping}
\chi_\alpha=\chi_\alpha\sum_{i} u_{v,i}(1-z_{v_i})
\end{equation}
By multiplying these identities over all $v\in I_\alpha$, then expanding modulo Booleanity, PC derives
\begin{equation}\label{eq:T-expand}
\chi_\alpha - \sum_{\sigma} \chi_\alpha U_\sigma W_\sigma 
\end{equation}
where the sum is over all functions $\sigma$ that choose one child $\sigma(v)$ per $v\in I_\alpha$, and
\begin{equation}\label{eq:T-UW}
U_\sigma=\prod_{v\in I_\alpha}u_{v,\sigma(v)}, \qquad W_\sigma=\prod_{v\in I_\alpha}(1-z_{v_{\sigma(v)}}).
\end{equation}
This derivation has degree at most $3t_O+O(1)$.
Below, for each $\sigma$, we derive  
\begin{equation}\label{eq:T-goal}
\chi_\alpha W_\sigma R
\end{equation} 
in degree $\deg R+2t_O+t_M (p-1)^2+O(1)$. If so, multiplying \eqref{eq:T-goal} by $U_\sigma$ and adding with the $R$-multiple of \eqref{eq:T-expand}, we get $\chi_\alpha R$; then summing them we get $R$, overall in degree
\begin{equation}\label{eq:T-degree}
\deg R + 3 t_O + t_M (p-1)^2 +O(1).
\end{equation}

\emph{Derivation of \eqref{eq:T-goal}.} 
Semantically, $\chi_\alpha W_\sigma=1$ forces every OR gate $v$ to have value $\alpha_v$ and each $\sigma(v)$ to be TRUE. Syntactically, we can derive $\chi_\alpha (z_v - \alpha_v)$ for each OR gate $v$ and derive $W_\sigma z_{v_{\sigma(v)}}$ for each $v\in I_\alpha$, using only Booleanity polynomials.
At other OR gates $v\notin I_\alpha$, $\chi_\alpha=1$ forces every child $v_i$ to be FALSE, i.e., we can derive $\chi_\alpha (1-z_{v_i})$ as a multiple of \eqref{eq:T-or-err} for each children $v_i$. 

If these forced values are inconsistent at a gate, then PC derives $\chi_\alpha W_\sigma$ from the two derived polynomials at that gate. Otherwise, the forced values define a partial assignment $\rho$ on $(z_v)_{v\in T}$. This means for every
variable fixed by $\rho$, PC derives
\begin{equation}\label{eq:rho}
\chi_\alpha W_\sigma(z_j-\rho(z_j))
\end{equation}
within degree $\deg (\chi_\alpha W_\sigma)+O(1)$.
Consequently, for any polynomial $h$, PC derives
\begin{equation}\label{eq:rho-restriction}
\chi_\alpha W_\sigma(h-h|_\rho)
\end{equation}
within degree $\deg (\chi_\alpha W_\sigma)+\deg h+O(1)$. In particular, PC derives, using $\mathsf G(T)$:
\begin{itemize}
\item $\chi_\alpha W_\sigma(R-R|_\rho)$;
\item $\chi_\alpha W_\sigma (h_j|_\rho)$, where $h_1,\cdots,h_{t_M}$ are the MOD-gate polynomials \eqref{eq:T-mod} in $\mathsf G(T)$ all having degree at most $p-1$.
\end{itemize}
These derivations are in degree $\deg (\chi_\alpha W_\sigma)+\max\{\deg R,p-1\}+O(1)$.

Every Boolean assignment to the remaining $\{z_v\}$-variables that makes $h_1|_\rho=\cdots=h_t|_\rho=0$ extends, via $\rho$, to a complete $\{z_v\}$-assignment that respects all gate computations. Hence $R|_\rho$ vanishes on these assignments. It follows that
\[
R|_\rho\prod_{j=1}^{t_M}\left(1 - (h_j|_\rho)^{p-1}\right)
\]
vanishes on the entire Boolean cube. By multilinear reduction, this polynomial is derivable from the Boolean axioms in degree at most
$\deg R+t_M (p-1)^2+1$. 
Multiplying this by $\chi_\alpha W_\sigma$ and using the already derived $\chi_\alpha W_\sigma (h_j|_\rho)$, we obtain $\chi_\alpha W_\sigma R|_\rho$, and hence $\chi_\alpha W_\sigma R$, in degree 
\begin{equation*}
2t_O+\deg R+t_M(p-1)^2+O(1). \qedhere
\end{equation*}
\end{proof}

We now simulate $\Pi$. The initial clauses in $F$ are handled by~\eqref{eq:syntactic-initial}, in degree $K+w(F)+O_p(\Delta_\ell)$.

To simulate the logical axioms and inferences in $\Pi$, we apply Lemma~\ref{lem:gate-sim}. 

Each axiom schema is a template with constantly many OR and MOD gates (even if its fan-ins are unbounded). 
When used in $\Pi$, it is instantiated as an initial formula $A$ whose leaf and gate approximators have all been constructed; note that identical leaf names are instantiated by the same formula in $\mathcal S_\Pi$, so they share an approximator by the construction.  
By applying the lemma with $R=z_{\text{A}}$, then substituting $z_v=P_v$ and $u_{v,i}=U_{v,i}$ from \eqref{eq:or-telescoping} in the resulting derivation, and finally multiplying every step by $f$, we get a derivation of $fP_A$ in degree $K+O_p(\Delta_\ell)$. The substituted Boolean axioms follow from the original Boolean axioms, the negation and modular relations hold exactly by the definitions of $P_A$, and the two OR relations follow from \eqref{eq:or-polys} and \eqref{eq:or-telescoping}. Therefore, this shows $fP_A\in\C_D$.

For an inference $\frac{A_1,\ldots,A_s}{B}$ ($s\le 2$), take the joint template and let $R=z_B\prod_{i=1}^s(1-z_{A_i})$.
It vanishes on every gate-consistent assignment by soundness of the rule. The template has constantly many OR and MOD gates; in particular, merging/unmerging uses only 3 OR gates. Applying the lemma thus gives $fP_B\prod_{i=1}^s(1-P_{A_i})\in\C_D$. By the inductive hypothesis, $fP_{A_i}\in\C_D$, so we obtain $fP_B\in\C_D$ using the identity
\[ fP_B = fP_B\prod_{i=1}^s(1-P_{A_i}) +\sum_{i=1}^s(fP_{A_i})P_B \prod_{j<i}(1-P_{A_j}),\]
again in degree $K+O_p(\Delta_\ell)$.

Thus, by \eqref{eq:sim-reserve}, every proof line $A$ satisfies $fP_A\in\C_D$. The final empty disjunction has $P_A=1$, so $f\in\C_D$, contradicting the choice of $f$.
\end{proof}

\end{document}